\documentclass[12pt, reqno]{amsart}
\usepackage{amsmath,amsthm,amssymb}
\DeclareMathOperator*{\argmax}{arg\,max}

\def\R{\mathbb{R}}

\newtheorem{thm}{Theorem}[section]
\newtheorem*{mainthm}{Main Result}
\newtheorem{cor}[thm]{Corollary}
\newtheorem{lm}[thm]{Lemma}
\newtheorem{ft}[thm]{Fact}

\theoremstyle{remark}
\newtheorem{rk}[thm]{Remark}
\newtheorem{st}[thm]{Statement}
\newtheorem{as}[thm]{Assumption}
\theoremstyle{definition}
\newtheorem{df}[thm]{Definition}
\newtheorem{eg}[thm]{Example}
\usepackage[hidelinks]{hyperref}
\usepackage{orcidlink}
\title{Existence and structure of Nash equilibria for supermodular games}
\author{Lu YU\,\orcidlink{0000-0001-6154-4229}}
\address{Université Paris 1 Panthéon-Sorbonne, UMR 8074, Centre d'Economie de la Sorbonne, Paris, France}\email{yulumaths@gmail.com}

\date{\today}
\begin{document} \begin{abstract}
		Two theorems announced by Topkis about the topological description of sublattices are proved. They are applied to extend some classical results concerning the existence and the order structure of Nash equilibria of certain supermodular games, with some problems in Zhou's proof corrected.
	\end{abstract}
	
	\maketitle
\section{Introduction}\label{sec:Intro}
The purpose of this note is to generalize the classical 	results concerning the existence and the structure of Nash equilibria  in supermodular games. Let $G=(N,\{S_i\}_{i\in N}, S,\{f_i\}_{i\in N})$ be a supermodular game (see Definition \ref{dfsupgame}), where $N$ is the set of players, $S_i$ (resp. $f_i$) is the strategy set (resp.  payoff function) of player $i$, and $S \subset \prod_{i \in N} S_i$ is the set of feasible strategy profiles.  Importantly, we do not require necessarily $S=\prod_{i \in N} S_i$. Such games are sometimes called generalized games in the literature (see, e.g., \cite[p.469]{braouezec2023economic}), and are very useful in economics (e.g., in general equilibrium). In particular, this class encompasses situations where the strategy choices of some player are constrained by the choices of other players. 	

A famous  result of Topkis \cite[Theorem 3.1]{topkis1979equilibrium} proved that the set of Nash equilibria of a supermodular game is nonempty and contains a maximum and a minimum, with many applications in game theory.  Later,  Zhou \cite[Theorem 2]{zhou1994set}  proved that the set of Nash equilibria of such games is in fact a complete lattice, generalizing another result of Vives  \cite[Theorem 4.1]{vives1990nash}. In Section \ref{sec:Intro}, the player set $N$ is assumed to be finite, as it is assumed in Facts \ref{ft:Zhougame} and \ref{ft:Toplisgame}, but we will relax this assumption in Theorems \ref{thm:inf} and \ref{thm:inf2}. 
\begin{ft}[Zhou]\label{ft:Zhougame}
	Assume that 	the game is in product form, \textit{i.e.}, $S=\prod_iS_i$. Assume that for each $i\in N$, $S_i$ is compact for a topology $\tau_i$ finer than the interval topology, and the payoff function $f_i(\cdot,x_{-i}):S_i\to \R$ is upper semicontinuous relative to $\tau_i$ for each $x_{-i}\in \prod_{j\neq i}S_j$. Then the set of Nash equilibria is a nonempty complete lattice.
\end{ft}
Based on Zhou's work, Topkis \cite[Theorem 4.2.1]{topkis1998supermodularity} enhanced his previous result as follows.
\begin{ft}[Topkis]\label{ft:Toplisgame}
	Assume that  each  $S_i$ is a sublattice of some Euclidean space $\R^{m_i}$ endowed with the subspace topology, and the set
	$S$ of feasible joint strategies is nonempty and compact in the product topology on $\prod_{i\in N}S_i$. Assume that for each $x_{-i}\in S_{-i}$ and each $i\in N$, the payoff function $f_i(\cdot,x_{-i}):\{x_i\in S_i:(x_i,x_{-i})\in S\}\to \R$ is upper semicontinuous relative to the subspace topology. Then the set of Nash equilibria is a nonempty complete lattice. In particular, a greatest and a least equilibrium point exist.
\end{ft}
Hypotheses in Facts \ref{ft:Zhougame} and \ref{ft:Toplisgame} are at the same time topological and order theoretical. The existence of equilibria for supermodular games (satisfying certain purely order-theoretical assumptions) is given in Veinott's unpublished work \cite[Ch.~10, Thm.~2]{veinott1992lattice}. As far as we are aware, Veinott's result does not contain Fact \ref{ft:Zhougame} or Fact \ref{ft:Toplisgame}.

Fact \ref{ft:Toplisgame}  restricts its conclusion to  games with strategies in Euclidean spaces because the proof  uses  \cite[Theorem 2.3.1]{topkis1998supermodularity}. By contrast, Fact \ref{ft:Zhougame} removes this restriction but is valid only for games in product form.  Moreover, as pointed out in Section \ref{secmain}, the proof in \cite{zhou1994set} is in fact insufficient for the stated generality. Here we generalize Facts \ref{ft:Zhougame} and \ref{ft:Toplisgame} by allowing more flexible strategy sets.

\begin{mainthm}[Theorem \ref{thm:main}] For every $i\in N$, let $\tau_i$ be a  topology on $S_i$ finer than the interval topology. Assume that 	 the sublattice $S \subset \prod_i S_i$ is closed and compact in the product topology of the $\tau_i$. If for	every $i\in N$, every $x\in S$, the function $f_i(\cdot,x_{-i}): S_i(x_{-i})\to \R$ is upper semicontinuous, then the set  of Nash equilibria is a nonempty  complete lattice.  \end{mainthm}

The paper is organized as follows.
In Section \ref{sec:partialorder}, we give two theorems in lattice theory, useful for our main result. They were stated by Topkis without published proofs, so for the sake of completeness, we provide proofs.
In Section \ref{sec:game}, we recall standard definitions related to supermodular games.  In Section \ref{secmain}, we prove our main theorem, as a consequence of Theorem \ref{thm:Topkis2}. For games with infinitely many players, two results (Theorems \ref{thm:inf} and  \ref{thm:inf2}) are also proved.
\section{Partial order and lattice}\label{sec:partialorder}
In this section, we prove two results stated by Topkis \cite[p.31]{topkis1998supermodularity}. They  make our main result possible.

Throughout this note, a set with a partial order is called a poset. The products of posets are always endowed with the product relation. We recall the following standard definitions: 

\begin{df}[Lattice]\cite[p.13]{topkis1998supermodularity}
	A poset $L$ is called a lattice if for every $x,y\in L$, both $\sup_L\{x,y\}$ and $\inf_L\{x,y\}$ exist in $L$, classically denoted by $x\vee_Ly$ and $x\wedge_L y$ respectively. 
\end{df}
\begin{df}[Complete lattice]\cite[p.29]{topkis1998supermodularity}
	A poset $L$ is called a complete lattice if  every nonempty subset $A$ of $L$, both $\sup_L(A)$ and $\inf_L(A)$ exist. (Such an $L$ is necessarily a lattice.)
\end{df}
\begin{df}[Increasing difference]\label{dfincdiff}\cite[p.42]{topkis1998supermodularity}
	Suppose that $X$ and $T$ are two posets. Let $S$ be a subset of $X\times T$ and $f :S\to \R$ be a function.
	If for every $x\lneq x'\in X$, every $t\lneq t'\in T$ with $\{x,x'\}\times \{t,t'\}\subset S$ we have \[f(x',t)+f(x,t')\le f(x,t)+f(x',t'),\]
	
	then  we say $f$ has increasing differences relative to $X\times T$.
\end{df}

\begin{df}[Interval topology]\cite[p.29]{topkis1998supermodularity}
	Let $X$ be a poset. The topology on $X$ generated by the subbasis for closed subsets  comprised of closed intervals $(-\infty,x]$ and $[x,+\infty)$ (where $x$ runs over $X$)  is called the interval topology of $X$.
\end{df}

\begin{df}[Increasing correspondence]\cite[p.159]{milgrom1994monotone}
	Let $ T$ (resp. $L$) be a  poset (resp. lattice). A correspondence $\phi:T\twoheadrightarrow L$ is called increasing if for any $t\le t'$ in $T$, every $x\in \phi(t)$ and every $x'\in \phi(t')$, one has \[x\wedge x'\in\phi(t),\quad x\vee x'\in \phi(t').\]
\end{df}

\begin{df}\cite[p.13, p.29]{topkis1998supermodularity}
	Let	$P$ be a  poset, $S\subset P$ be a subset. If for each nonempty finite (resp. nonempty but possibly infinite) subset $A\subset S$, both $\sup_P(A)$ and $\inf_P(A)$ exist and belong to $S$, then $S$  is called a sublattice (resp. subcomplete sublattice)  of $P$.
\end{df}	
Subcomplete sublattices are called \emph{closed in the lattice theoretical sense} in \cite[p.296]{zhou1994set}.

Fact \ref{ft:Zhougame} has been proved by Zhou using lattice theoretic Statement \ref{st:Zhou}, which is  wrong. 
\begin{st}\label{st:Zhou}
	\hfill 
	\renewcommand\labelenumi{(\theenumi)}
	\begin{enumerate}
		\item \textup{(\cite[Remark 2]{zhou1994set})} Any sublattice of a complete lattice is closed in the lattice theoretical sense if and only if it is topologically closed in the interval topology.
		\item \textup{(the second paragraph of the proof  of \cite[Thm.~2]{zhou1994set})} In the interval topology of a complete lattice, a compact subset is closed.
\end{enumerate}\end{st}
A counterexample to both  points in Statement \ref{st:Zhou} is  Example \ref{egmistake}. An \emph{anti-chain} is a poset where any two distinct elements are incomparable.

\begin{eg}\label{egmistake}
	Let $L$ be the set obtained by adding two elements $m\neq M$ to  an infinite anti-chain $X$. We extend the partial order from $X$ to $L$ as follows. For every $x\in X$, we require that $m\le x\le M$. Then $L$ is a complete lattice. In fact, for every nonempty subset $S$ of $L$, one has $\inf_L(S)=m$ and $\sup_L(S)=M$. The interval topology of $L$ is the cofinite topology. Fix $x_0\in X$. Then $L\setminus\{x_0\}$ is a subcomplete sublattice of $L$. It is compact  but not closed in the interval topology of $L$. 
\end{eg}
The proof of Fact \ref{ft:Zhougame} also relies on  Zhou's generalization (Fact \ref{ft:TarskiZhou}) of Tarski's fixed point theorem \cite[Theorem 1]{tarski1955lattice}. For a self correspondence $Y:S\twoheadrightarrow S$ on a set $S$, let $\mathrm{Fix}(S)=\{x\in S:x\in Y(x)\}$. 
\begin{ft}[Tarski-Zhou, {\cite[Theorem 1]{zhou1994set}}]\label{ft:TarskiZhou}
	Let	$S$ be a complete lattice. Let $Y:S\twoheadrightarrow S$ be an increasing correspondence, such that for every $x\in S$, $Y(x)$ is a nonempty \emph{subcomplete} sublattice of $S$. Then 
	$\mathrm{Fix}(Y)$ is  a nonempty complete lattice.
\end{ft}
We shall prove two theorems announced without proof by Topkis. This requires Lemma \ref{excompare}.
\begin{lm}\label{excompare}
	\hfill
	\renewcommand\labelenumi{(\theenumi)}
	\begin{enumerate}
		\item\label{it:subint} Let $P$ be a poset and 
		$Q$ be a subcomplete sublattice of $P$. Let $\tau_P$ (resp. $\tau_Q$) be the interval topology of $P$ (resp. $Q$). Then  $\tau_Q=\tau_P|_Q$.
		\item\label{it:prodint} Let $\{P_i\}_{i\in I}$ be a family of posets, and let $\tau_i$ be the interval topology of $P_i$. Denote the product topology  of the topologies $\tau_i$ on $P:=\prod_iP_i$ by $\tau$. For each $i\in I$, let $L_i\subset P_i$ be a subcomplete sublattice. Then the restriction of $\tau$ to $\prod_iL_i$  is the interval topology of $\prod_iL_i$.
	\end{enumerate}
\end{lm}
\begin{proof}
	\hfill
	\renewcommand\labelenumi{(\theenumi)}
	\begin{enumerate}
		\item For every $a\in Q$, the interval $[a,+\infty)_Q=[a,+\infty)_P\cap Q$ is closed in $\tau_P|_Q$. Similarly, $(-\infty,a]_Q$ is closed in $\tau_P|_Q$. Therefore, the topology 
		$\tau_Q$ generated by these intervals is weaker than $\tau_P|_Q$.
		
		For every $y\in P$, we show that $[y,+\infty)_P\cap Q$ is closed in $\tau_Q$. 
		
		If this set is empty then it is trivially closed. If this set is nonempty, then as $Q$ is subcomplete in $P$, we know that $z:=\inf_P([y,+\infty)_P\cap Q)$ exists in $Q$. For any $x\in [y,+\infty)_P\cap Q$, we have $x\ge y$ and $x\ge z$.  So $z\ge y$. Therefore, $[y,+\infty)_P\cap Q=[z,+\infty)_Q$ is closed in $\tau_Q$.
		
		Similarly, we can show that $(-\infty,y]_P\cap Q$ is closed in $\tau_Q$. Thus, $\tau_Q\supset \tau_P|_Q$. 
		\item  As a property of product topology, the product  of   the topologies $\{\tau_i|_{L_i}\}_{i\in I}$ on $\prod_{i\in I}L_i$ is exactly the restriction $\tau|_{\prod_iL_i}$ of $\tau$. For every $i\in I$, by  Point (\ref{it:subint}), the interval topology $\sigma_i$ of $L_i$ equals $\tau_i|_{L_i}$. Moreover, the lattice $L_i$ is complete hence bounded. By \cite[Thm.~3]{alo1966topologies}, the interval topology of $\prod_iL_i$ is the product topology $\prod_i\sigma_i$.\end{enumerate}
\end{proof}

Theorems \ref{thm:Topkis} and \ref{thm:Topkis2} are stated in \cite[p.31]{topkis1998supermodularity}, whose proofs are referred  to  the unpublished work \cite{toplisunpublished}. They have applications to supermodular games. Theorem \ref{thm:Topkis} appears without proof in \cite[Footnote 3]{milgrom1994monotone}.   For the convenience of the reader, we include proofs.
\begin{thm}\label{thm:Topkis}
	Let	$L$ be a poset, and let $S\subset L$ be a  sublattice. Then $S$ is a subcomplete sublattice if and only if $S$ is compact in the interval topology of $L$.
\end{thm}

\begin{proof}
	First,  assume that $S$ is a subcomplete sublattice. Then $S$ is a complete lattice. By the Frink-Birkhoff theorem  (\cite[Thm.~20, p.250]{birkhoff1940lattice}), $S$ is compact in its own interval topology. By Lemma \ref{excompare} (\ref{it:subint}), $S$ is compact in the interval topology of $L$.
	
	Conversely,  assume that $S$ is compact in the interval topology of $L$. Then $S$ is compact in its own interval topology. By the Frink-Birkhoff theorem, the lattice $S$ is  complete. In particular, $\max(S)$ exists.
	
	For every nonempty subset $A$ of $S$, we  show that $\sup_L(A)$ exists and is in $S$.  Let $B=\{x\in L:x\ge a,\forall a\in A\}$. Then $\max(S)\in B$. In particular, $B$ is nonempty. 
	
	Consider a nonempty family of subsets $\{S\cap [a,b]_L\}_{a\in A,b\in B}$ of $S$. For every $a\in A$, every $b\in B$, the interval $[a,b]_L$ is closed in $L$, so $S\cap[a,b]_L$ is closed in $S$ in the interval topology of $L$.  Consider any nonempty finite subfamily: $\{S\cap[a_i,b_i]_L\}_{i=1}^m$, where $m\ge 1$ is an integer,  $a_i\in A$, and $b_i\in B$. As $S$ is a sublattice of $L$,  \[c:=\sup_L\{a_i:1\le i\le m\}=\sup_S\{a_i:1\le i\le m\}\] exists in $S$. By definition, $c\ge a_i$ for all $1\le i\le m$. For every $1\le i\le m$, as $b_i\in B$ we have $b_i\ge a_j$ for all $1\le j\le m$, so $b_i\ge c$.
	Therefore, $c$ is in the intersection of this finite subfamily. 
	
	By the compactness of $S$, the intersection of the whole family  contains an element $z$. For every $a\in A$, every $b\in B$, one has $z\in S\cap [a,b]_L$, so $a\le z$. Therefore, $z\in B$ is the least element of $B$. Thus, $\sup_L(A)=\min(B)=z$ exists and is in $S$.
	
	Similarly,  $\inf_L(A)$ exists and is in $S$. Therefore, the sublattice $S$ is  subcomplete.
\end{proof}
\begin{rk}When $S=L$,  Theorem \ref{thm:Topkis}  reduces to the Frink-Birkhoff Theorem.\end{rk}
\begin{thm}\label{thm:Topkis2}
	Let $\{P_i\}_{i\in N}$ be a nonempty family of posets, and let $S$ be a nonempty sublattice of $P:=\prod_iP_i$. Then  $S$ is a subcomplete sublattice of $P$ if and only if $S$ is compact in the product topology  $\tau$ of the individual interval topology $\tau_i$ of each $P_i$.
\end{thm}
\begin{proof}
	First, assume that $S\subset P$ is a subcomplete sublattice. For each $i\in N$, let $L_i$ the projection of $S$ on $P_i$. Then $L_i$ is a subcomplete sublattice of $P_i$. 	
	
	Also, $S$ is a subcomplete sublattice of $\prod_iL_i$. By Theorem \ref{thm:Topkis}, $S$ is compact in the interval topology $\sigma$ of $\prod_iL_i$. By Lemma \ref{excompare} (\ref{it:prodint}), $\sigma$ is the restriction of $\tau$. Therefore, $S$ is compact in the topology $\tau$.
	
	Conversely, assume that $S$ is compact in $\tau$. The interval topology of $P$ is weaker than $\tau$, 
	so $S$ is also compact in the interval topology of $P$. By Theorem \ref{thm:Topkis}, $S$ is a subcomplete sublattice of $P$.
\end{proof}
\begin{rk}In Theorem \ref{thm:Topkis2}, if  $N$ is a singleton, then we recover Theorem \ref{thm:Topkis}. If $P_i=\R^1$ for all $i\in N$, then we recover \cite[Theorem 2.3.1]{topkis1998supermodularity}.\end{rk}

\section{Supermodular game}\label{sec:game}
In  Section \ref{sec:game}, standard definitions and basic facts related to supermodular games are recalled.

\begin{df}\label{df:ncg}
	A noncooperative game $G=(N,\{S_i\}_{i \in N},S,\{f_i\}_{i \in N})$ is the following data:
	\begin{itemize}
		\item a nonempty  set of players $N$;
		\item for each $i\in N$, a nonempty set $S_i$ of the strategies of player $i$;
		\item a nonempty subset $S\subset \prod_{i\in N}S_i$ of feasible joint strategies satisfying that each projection $S\to S_i$ is surjective;
		\item for each player $i\in N$, a payoff function $f_i: S\to \R$.
\end{itemize}\end{df}
\begin{rk}Definition \ref{df:ncg} is slightly more general than that in \cite[p.176]{topkis1998supermodularity}. According to \cite[p.178]{topkis1998supermodularity}, in real applications to game theory, a profile of  strategies $x\in \prod_iS_i$  is feasible for each player $i\in N$ but may not be allowed as a joint strategy. So,  Definition \ref{df:ncg} is not limited to  the case of $S=\prod_iS_i$.\end{rk}

For a noncooperative game $G=(N,\{S_i\}_{i \in N},S,\{f_i\}_{i \in N})$, every joint strategy $x\in S$ and every player $i\in N$, let $x_{-i}\in\prod_{j\neq i}S_j$ denote the strategies of all players in $N$ except player $i$ for $x$. Put \[S_i(x_{-i}):=\{y_i\in S_i: (y_i,x_{-i})\in S\}.\] It is nonempty, since	$x_i\in S_i(x_{-i})$. For every player $i\in N$, let $S_{-i}$ be the image of $S$ under the projection $S\to \prod_{j\neq i}S_j$.

For $x\in S$, put \[S(x)=(\prod_{i\in N}S_i(x_{-i}))\cap S.\] Since $x\in S(x)$, $S(x)$ is nonempty. Put \[\Delta:=\{(y,x)\in S\times S:y\in S(x)\}.\] Then $\Delta$ is a sublattice of $S\times S$.

\begin{df}\cite[p.178]{topkis1998supermodularity}
	A feasible joint strategy $x\in  S$ is an equilibrium  if $f_i(y_i,x_{-i})\le f_i(x)$ for every $i\in N$ and every $y_i\in S_i(x_{-i})$. The subset of $S$ of all equilibria is denoted by $E$.
\end{df}
Set \begin{equation}\label{eq:Fi}
	F_i:=\{x\in S:f_i(y_i,x_{-i})\le f_i(x),\quad \forall y_i\in S_i(x_{-i})\}.
\end{equation}
Then \begin{equation}\label{eq:E=FN}
	E=\cap_{i\in N}F_i.
\end{equation}

\begin{df}\cite[pp.177--178]{topkis1998supermodularity}\label{dfbestjoint}
	For each player $i\in N$, the (individual) best response correspondence  $Y_i:S\twoheadrightarrow S_i$ is defined by \[Y_i(x)=\argmax_{y_i\in S_i(x_{-i})}f_i(y_i,x_{-i}).\] (Since $Y_i(x)$ depends only on $x_{-i}$, by abuse of notations, we sometimes write $Y_i(x_{-i})$ to emphasis this point.) 	More generally, when $I$ is a \emph{finite} nonempty subset of $N$,		the best partial response correspondence $Y_I:S\twoheadrightarrow S$  relative to $I$ is defined as \[Y_I(x)=\argmax_{y\in S(x)}g_I(y,x),\] where for every $x\in S$, the function $g_I(\cdot,x):S(x)\to \R$ is defined by \[g_I(y,x)=\sum_{i\in I}f_i(y_i,x_{-i}).\] When $N$ is finite, the best joint response correspondence is $Y=Y_N$.
\end{df}Functions similar to $g_I$ are   introduced in \cite[(1), p.808]{nikaido1955note}. They are used in \cite[Theorem 1]{rosen1965existence} to study the equilibria for games in non-product spaces.

The correspondence $Y_I$ shall play a role in the case of infinitely many players. When $I=\{i\}$ is a singleton, one has \[Y_I(x)=(Y_i(x)\times \prod_{j\neq i}S_j)\cap S(x).\] The classical best response $R:S\twoheadrightarrow S$ is defined as $R(x)=(\prod_{i\in N}Y_i(x_{-i}))\cap S$. 
Here is a comparison of the two. \begin{lm}\label{exRY}
	For every $x\in S$, every finite nonempty subset $I$ of $N$,  one has $R(x)\subset Y_I(x)$.
\end{lm}
\begin{proof}
	For every $y\in R(x)$,  $y'\in S(x)$, and $i\in I$, one has $y_i\in Y_i(x_{-i})$, thus $f_i(y_i,x_{-i})\ge f_i(y'_i,x_{-i})$. Taking sum over $i\in I$, we get $g_I(y,x)\ge g_I(y',x)$, so $y\in Y_I(x)$.
\end{proof}

The correspondences $Y$ and $R$ are used to convert the problem concerning equilibria into a fixed point problem.
\begin{ft}\label{equi}
	We have $E=\mathrm{Fix}(R)$. 
\end{ft}
\begin{lm}\label{lm:FixY}
	For every finite nonempty subset $I$ of $N$, $\mathrm{Fix}(Y_I)=\cap_{i\in I}F_i$. In particular, if $N$ is finite, then $\mathrm{Fix}(Y_N)=E$.
\end{lm}
\begin{proof} 
	If $x\in S\setminus\cap_{i\in I}F_i$, then there is $i_0\in I$ with $x\notin F_{i_0}$.  In particular, there is $y_{i_0}\in S_{i_0}(x_{-i_0})$ with $f_{i_0}(x)<f_{i_0}(x')$, where $x'_{i_0}=y_{i_0}$ and $x'_i=x_i$ for all $i\neq i_0$. Consequently, we have  $x'\in S(x)$. Now, for every $i\neq i_0$, $f_i(x'_i,x_{-i})=f_i(x)$, so $g_I(x',x)>g_I(x,x)$, which implies \[x\notin \argmax_{y\in S(x)}g_I(y,x)=Y_I(x),\] and finally, $x\notin \mathrm{Fix}(Y_I)$.  
	
	Conversely, if $x\in \cap_{i\in I}F_i$, then for every $y\in S(x)$, every $i\in I$, we have $y_i\in S_i(x_{-i})$. So $f_i(y_i,x_{-i})\le f_i(x)$. Summing over $i\in I$, we get $g_I(y,x)\le g_I(x,x)$. Therefore, \[x\in \argmax_{y\in S(x)}g_I(y,x)=Y_I(x),\] \textit{i.e.}, $x\in \mathrm{Fix}(Y_I)$.
	
	The second part follows from the first one and (\ref{eq:E=FN}).
\end{proof}
\begin{df}\label{dfsupgame}
	A noncooperative game $G=(N,\{S_i\}_{i \in N},S,\{f_i\}_{i \in N})$ is called supermodular if
	\begin{itemize}

		\item	each $S_i$ is a lattice and $S$ is a sublattice of $\prod_iS_i$;
		\item for  every $i\in N$, every $x_{-i}\in S_{-i}$, the function $f_i(\cdot,x_{-i}):S_i(x_{-i})\to \R$ is  supermodular; 
		\item 	for every $i\in N$, the function $f_i$ has increasing difference relative to the subset $S$ of $S_i\times (\prod_{j\neq i}S_j)$ in the sense of Definition \ref{dfincdiff}.
	\end{itemize}
\end{df}
\begin{rk} Definition \ref{dfsupgame} extends the definition in \cite[p.178]{topkis1998supermodularity} (which requires further that each $S_i$ is a sublattice of some Euclidean space), and the definition in \cite[p.298]{zhou1994set} (which requires that $S=\prod_iS_i$ is in the product form).
	
	By \cite[Lemma 2.2.3, p.17]{topkis1998supermodularity}, for every $i\in N$ and every $x\in S$, $S_i(x_{-i})$ is a sublattice of $S_i$. So we are allowed to talk about the supermodularity of the function $f_i(\cdot,x_{-i}):S_i(x_{-i})\to \R$. Moreover, for every $x\in S$, $\prod_iS_i(x_{-i})$ is a sublattice of $\prod_iS_i$. Therefore, $S(x)$	 is a sublattice of $\prod_{i\in N}S_i$ and of $S$. \end{rk}

\begin{lm}\label{ex2.4.5}For a supermodular game as above, 
	\renewcommand\labelenumi{(\theenumi)}
	\begin{enumerate}
		\item\label{it:StoSi} for every $i\in N$, the correspondence \[S\twoheadrightarrow S_i,\quad x\mapsto S_i(x_{-i})\] is increasing; 
		\item \label{it:S(x)inc}	the correspondence \[S\twoheadrightarrow S,\quad x\mapsto S(x)\] is increasing.
	\end{enumerate}
\end{lm}
\begin{proof}
	For every $x\le x'$ in $S$, 
	\renewcommand\labelenumi{(\theenumi)}
	\begin{enumerate}
		\item for every $y_i\in S_i(x_{-i})$ and $y_i'\in S_i(x_{-i}')$, both	$(y_i,x_{-i})$ and $(y'_i,x'_{-i})$ belong to $S$. Since $S$ is a sublattice of $\prod_{i\in N}S_i$, we have \[(y_i,x_{-i})\wedge (y'_i,x'_{-i})=(y_i\wedge_iy'_i,x_{-i})\in S.\] Thus, $y_i\wedge_iy'_i\in S_i(x_{-i})$. Similarly, $y_i\vee y'_i\in S_i(x'_{-i})$. 
		\item For every $y\in S(x)$ and every $y'\in S(x')$, since $S$ is a sublattice of $\prod_{i\in N}S_i$, we have 
		$y\wedge y'\in S$. By Point (\ref{it:StoSi}), for every $i\in N$, $(y\wedge y')_i=y_i\wedge y'_i\in S_i(x_{-i})$, so $y\wedge y'\in S(x)$. Similarly, $y\vee y'\in S(x')$. 	  
\end{enumerate}\end{proof}

\section{Main results}\label{secmain}
As recalled in Section \ref{sec:partialorder}, Zhou's proof of Fact \ref{ft:Zhougame} rests on Statement \ref{st:Zhou} which is false in general. In this section, we correct Zhou's proof and generalize Fact \ref{ft:Zhougame} in  Theorem \ref{thm:main}.

In Section \ref{secmain}, fix a supermodular game $G=(N,\{S_i\}_{i \in N},S,\{f_i\}_{i \in N})$ and for every $i \in N$, consider a topology $\tau_i$ on $S_i$. Unless otherwise specified, we use the product topology on $\prod_i S_i$.  

In the whole section,	\textbf{we assume Assumption \ref{as}}. 

\begin{as}\label{as}\hfill\begin{itemize}
		\item For every $i\in N$, the topology $\tau_i$ is finer than the interval topology of $S_i$; 
		\item The sublattice $S \subset \prod_i S_i$ is closed compact in this topology;
		\item for	every $i\in N$, every $x\in S$, the function $f_i(\cdot,x_{-i}): S_i(x_{-i})\to \R$ is upper semicontinuous. 
\end{itemize}\end{as}

We now state the main theorem of this note.

\begin{thm}\label{thm:main}	
	If $I$ is a finite nonempty subset of $N$, then $\mathrm{Fix}(Y_I)$ is a nonempty  complete lattice. In particular,	if moreover $N$ is  finite, then  $E$  is a nonempty  complete lattice.
\end{thm}

First, we prove several lemmata that we use in the proof of Theorem \ref{thm:main}.
\begin{lm}\label{exSi(x_{-i})}For every $j\in N$,
	\renewcommand\labelenumi{(\theenumi)}
	\begin{enumerate}
		\item $S_j$ is a $\tau_j$-compact complete lattice;
		\item for every $x\in S$, $S_j(x_{-j})$ is compact closed in $S_j$. The subset $S(x)$ is closed compact in $\prod_iS_i$. 
	\end{enumerate}
\end{lm}
\begin{proof}
	\hfill
	\renewcommand\labelenumi{(\theenumi)}
	\begin{enumerate}
		
		\item Since the projection $S\to S_j$ is continuous surjective, $S_j$ is $\tau_j$-compact. By the assumption on $\tau_j$, $S_j$ is compact in its interval topology. The completeness of $S_j$ follows from the Frink-Birkhoff theorem.	
		\item 	Consider the continuous injection \[S_j\to \prod S_i,\quad y\mapsto (y,x_{-j}).\] The preimage $S_j(x_{-j})$ of $S$ is closed in $S_j$, thence also compact. Therefore, $\prod_iS_i(x_{-i})$ is compact by Tychonoff's theorem. Since $S$ is closed in $\prod_iS_i$, $S(x)$ is closed in $\prod_iS_i(x_{-i})$.  As $\prod_iS_i(x_{-i}) $ is closed in $\prod_iS_i$, the subset $S(x)$ is closed in $\prod_iS_i$.
	\end{enumerate}
\end{proof}

\begin{lm}\label{exScplt}The sublattice $S$ is subcomplete in $\prod_iS_i$.\end{lm}
\begin{proof}
	By assumption, $S$ is compact in the product of individual interval product of each $S_i$. We conclude by Theorem \ref{thm:Topkis2}.\end{proof}
\begin{lm}\label{exYi}
	For every $i\in N$, every $x\in S$,  the set $Y_i(x_{-i})$  is a nonempty closed compact subset of $S_i$.
\end{lm}
\begin{proof}By Lemma \ref{exSi(x_{-i})}, $S_i(x_{-i})$ is a nonempty compact space.
	By assumption, the function $f_i(\cdot,x_{-i}):S_i(x_{-i})\to \R$ is upper semicontinuous. Then \[Y_i(x)=\argmax_{y_i\in S_i(x_{-i})}f_i(y_i,x_{-i})\] is a nonempty closed compact subset of $S_i(x_{-i})$. 
\end{proof}
\begin{lm}\label{exYiinc}
	For every $i\in N$, the correspondence $Y_i:S\twoheadrightarrow S_i$ is increasing. In particular, for every $x\in S$, $Y_i(x)$ is a sublattice of $S_i$. 
\end{lm}
\begin{proof}
	For every $t\in S$, the function $f_i(\cdot,t_{-i}):S_i\to \R$ is supermodular. Let $B=\{(x_i,t)\in S_i\times S|(x_i,t_{-i})\in S\}$. The function \[B\to \R,\quad (x,t)\mapsto f_i(x_i,t_{-i})\] has increasing difference relative to $B\subset S_i\times S$. We conclude by Lemma \ref{exYi} and Topkis' monotonicity theorem \cite[Theorem 2.8.1]{topkis1998supermodularity}.
\end{proof}
\begin{lm}\label{Y}
	For every finite nonempty subset $I$ of $N$:
	\renewcommand\labelenumi{(\theenumi)}
	\begin{enumerate}
		\item\label{it:YI} The set $Y_I(x)$ is a nonempty closed \emph{compact} subset of $\prod_iS_i$ for every $x\in S$.
		\item\label{it:YIinc} The  correspondence $Y_I:S\twoheadrightarrow S$ is increasing. In particular, for every $x\in S$, $Y_I(x)$ is a sublattice of $S$ and of $\prod_iS_i$.
		\item For every $x\in S$, the sublattice $Y_I(x)$ is \emph{subcomplete} in $S$. In particular, both $\max Y_I(x)$ and $\min Y_I(x)$ exist.
	\end{enumerate}
\end{lm}
\begin{proof}
	\hfill
	\renewcommand\labelenumi{(\theenumi)}
	\begin{enumerate}
		\item  By assumption, for every $i\in N$, the function $f_i(\cdot,x_{-i}):S(x)\to \R$ is upper semicontinuous. As the sum of finitely many upper semicontinuous functions, the function $g_I(\cdot,x):S(x)\to \R$ is upper semicontinuous. By Lemma \ref{exSi(x_{-i})}, $S(x)$ is nonempty compact closed subset of $\prod_iS_i$. Then $Y_I(x)$ is a nonempty closed compact subset of $S(x)$.
		\item 	The correspondence \[S\twoheadrightarrow S,\quad x\mapsto S(x)\] is increasing by Lemma \ref{ex2.4.5} (\ref{it:S(x)inc}). For every $i\in N$, every $x\in S$, the function \[S(x)\to \R,\quad y\mapsto f_i(y_i,x_{-i})\] is supermodular by Definition \ref{dfsupgame}. As the sum of finitely many supermodular functions, the function $g_I(\cdot,x):S(x)\to \R$ is also supermodular.  By Definition \ref{dfsupgame}, the function \[\Delta\to \R,\quad (y,x)\mapsto f_i(y_i,x_{-i})\] has increasing difference relative to $\Delta\subset S_i\times (\prod_{j\neq i}S_j)$. As the sum of finitely many such functions, $g_I:\Delta\to \R$ also has increasing difference relative to $\Delta\subset S_i\times (\prod_{j\neq i}S_j)$. By Point (\ref{it:YI}) and \cite[Theorem 2.8.1, p.76]{topkis1998supermodularity}, the correspondence $Y_I$ is increasing. 
		\item 
		By Point (\ref{it:YI}) and Theorem \ref{thm:Topkis2}, $Y(x)$ is a subcomplete sublattice of $\prod_iS_i$, hence also a subcomplete sublattice of $S$. 
\end{enumerate}	\end{proof}

Now, we prove the main theorem.
\begin{proof}[Proof of Theorem \ref{thm:main}] 
	We apply  Fact \ref{ft:TarskiZhou}. For this, we check the conditions in order. By Lemma \ref{exScplt},  the lattice $S$ is complete. By Lemma \ref{Y},   the correspondence $Y_I$ is increasing and for every $x\in S$, the set $Y_I(x)$ is a nonempty subcomplete sublattice of $S$.  Now by Fact \ref{ft:TarskiZhou}, $\mathrm{Fix}(Y_I)$ is a nonempty complete lattice. 
	The second statement follows by  Lemma \ref{lm:FixY}.
\end{proof}

\begin{rk} In Theorem \ref{thm:main}, when each $(S_i,\tau_i)$ is a sublattice of some Euclidean space,   we recover Fact \ref{ft:Toplisgame}. In another direction, if we take $S$ to be the full $\prod_iS_i$, then we obtain  Fact \ref{ft:Zhougame}. \end{rk}
A noncooperative game is called \emph{finite} if the player set and the strategy set of every player  are finite.
\begin{cor}\label{cor:fini}
	For every finite supermodular game, the set of Nash equilibria is a nonempty  lattice.
\end{cor}
\begin{proof}
	For every $i\in N$, let $\tau_i$  be the discrete topology of $S_i$. Then Assumption \ref{as} is satisfied. We conclude by Theorem \ref{thm:main}. 	\end{proof}
\begin{rk}Although Corollary \ref{cor:fini} seems fundamental,	as far as we know  it is not explicitly stated in the literature. A special case (with every $S_i$ is a chain  and $S=\prod_iS_i$, as remarked in \cite[footnote on p.2]{takahashi2002pure}) is in \cite[Theorem 3]{friedman2001learning}.\end{rk}

Theorems \ref{thm:inf} and \ref{thm:inf2} are two results in the case of infinitely many players, but with some extra restrictions. 

\begin{thm}\label{thm:inf}
	If  moreover $S=\prod_i S_i$, then  $E$  is a nonempty  complete lattice.
\end{thm}

\begin{proof}By Lemma \ref{exYi}, for every $i\in N$, every $x\in S$, the set $Y_i(x_{-i})$ is nonempty and closed in $S_i$. Then  $R(x)=\prod_iY_i(x_{-i})$ is nonempty and closed in $S=\prod_iS_i$. As $S$ is compact, $R(x)$ is compact in the product topology, and hence compact in the interval topology of $S$. By Lemma \ref{exScplt}, the lattice $S$ is complete. On account of Lemma \ref{exYiinc}, for every $i\in N$, the correspondence $Y_i:S\twoheadrightarrow S_i$ is increasing. Therefore the correspondence $R$ is increasing. In particular, for every $x\in S$, $R(x)$ is a sublattice of $S$. By Theorem \ref{thm:Topkis}, $R(x)$ is a subcomplete sublattice of $S$. The proof is  completed  by  Facts \ref{ft:TarskiZhou} and \ref{equi}.
\end{proof}
If the game is not in the product form, we have only an existence result. 

\begin{thm}\label{thm:inf2}
	Assume further that  for every $i\in N$, the function $f_i:S\to \R$ is continuous and the correspondence \[\phi_i:S_{-i}\twoheadrightarrow S_i,\quad x_{-i}\mapsto S_i(x_{-i})\] is lower semicontinuous. Then $E$ is nonempty.
\end{thm}
\begin{proof}
	For every $i\in N$, we show that the $F_i$ defined in (\ref{eq:Fi}) is closed in $S$.
	
	In fact, take any $s\in S\setminus F_i$, there is $y_i\in S_i(s_{-i})$ such that $f_i(s)<f_i(y_i,s_{-i})$. Take $a\in (f_i(s),f_i(y_i,s_{-i}))\subset \R$. By the continuity of $f_i:S\to \R$, the set $W:=\{x\in S:f_i(x)>a\}$ is an open neighborhood of $(y_i,s_{-i})\in S$. By the definition of product topology on $S_i\times S_{-i}$, there exist open subsets $U\subset S_i$ and $V\subset S_{-i}$ with $(y_i,s_{-i})\in (U\times V)\cap S\subset W$. Then $y_i\in U$ and $s_{-i}\in V$.

	Set \[O:=\{x_{-i}\in S_{-i}:S_i(x_{-i})\cap U\neq\emptyset\}.\] By the lower semicontinuity of $\phi_i$, $O$ is open in $S_{-i}$. As $y_i\in U\cap S_i(s_{-i})$, one has $s_{-i}\in O$. By the continuity of $f_i$ on $S$ again, the subset \[A=(S_i\times (V\cap O))\cap \{x\in S:f_i(x)<a\}\] is open in $S$ and contains $s$. For every $x\in A$, one has  $f_i(x)<a$ and $x_{-i}\in O\cap V$, so there is $y'_i\in S_i(x_{-i})\cap U$. Then \[(y'_i,x_{-i})\in (U\times V)\cap S\subset W,\] so $f_i(y'_i,x_{-i})>a>f_i(x)$. Thus, $x\notin F_i$ for all $x\in A$, or equivalently $A\subset S\setminus F_i$. This shows the closedness of $F_i$.
	
	For every finite nonempty subset $I$ of $N$, by Lemma \ref{lm:FixY}, $\cap_{i\in I}F_i=\mathrm{Fix}(Y_I)$. It is nonempty by  \cite[Lemma 2.2.3, p.17]{topkis1998supermodularity} and Theorem \ref{thm:main}. By the compactness of $S$, the full intersection $E=\cap_{i\in N}F_i$ is nonempty.
\end{proof}
\begin{rk}A closely related result is \cite[Theorem 4.3.2, p.189]{topkis1998supermodularity}, which concerns the algorithmic aspect with $N$ being finite, while we study the existence of equilibria with $N$ being  infinite. 
\end{rk}
	\subsection*{Acknowledgments}
	The author would like to thank her advisor, Philippe Bich, for his helpful feedback and encouraging her to work on this project.
	\bibliography{ref.bib}
	\bibliographystyle{alpha}\end{document}